\documentclass[aps,pra,groupedaddress,superscriptaddress,twocolumn,nofootinbib,showpacs]{revtex4-1}

\usepackage{amsmath}
\usepackage{color}
\usepackage{amssymb}
\usepackage{amsthm}
\usepackage[hyperindex,breaklinks]{hyperref}

\newcommand{\CPLX}{\mathbb{C}}

\newcommand{\BC}{\mathcal{B}}

\newcommand{\FC}{\mathcal{F}}

\newcommand{\HC}{\mathcal{H}}

\newcommand{\ket}[1]{|#1\rangle}                  %ket
\newcommand{\bra}[1]{\left\langle #1 \right|}     %bra
\newcommand{\dyad}[2]{\ket{#1}\bra{#2}}           %dyad
\def\dya#1{|#1\rangle \langle#1|}

\long\def\ca#1\cb{} %Use for commenting out: \ca...\cb

\newtheorem{theorem}{Theorem}
\newtheorem{lemma}{Lemma}
\newtheorem{corollary}{Corollary}

\newtheorem{definition}{Definition}
\newtheorem{proposition}{Proposition}

\newcommand{\einv}{E_{inv}}

\begin{document}

%\linenumbers

\title{Multipartite Entanglement Evolution Under Separable Operations}

\author{Vlad Gheorghiu}
\email{vgheorgh@ucalgary.ca}
\author{Gilad Gour}
\email{gour@ucalgary.ca}
\affiliation{Institute for Quantum Information Science and Department of Mathematics and Statistics,\\University of Calgary, 2500 University Drive NW,\\Calgary, AB, T2N 1N4, Canada }

\date{Version of November 7, 2012}

\begin{abstract}
We study how multi-partite entanglement evolves under the paradigm of separable operations, which include the local operations and classical communication (LOCC) as a special case. We prove that the average ``decay" of entanglement induced by a separable operation is measure independent (among SL-invariant ones) and state independent: the ratio between the average output entanglement and the initial entanglement is solely a function of the separable operation, regardless of the input state and of the SL-invariant entanglement measure being used. We discuss the ``disentangling power" of a quantum channel and show that it exhibits a similar state invariance as the average entanglement decay ratio. Our article significantly extends the bipartite results of \cite{Horodecki:2010:QCS:2011451.2011452,Konrad:2008fk,PhysRevLett.101.170502} as well as the multi-partite one of \cite{PhysRevLett.105.190504}, all of the previous work being restricted to one-sided or particular noise models.
\end{abstract}

\maketitle

%\section{Introduction}\label{sct1}
\textit{Introduction.} 
Entanglement is a key ingredient in many quantum information protocols, such as  factoring \cite{SJC.26.1484}, teleportation  \cite{PhysRevLett.70.1895}, quantum dense coding \cite{PhysRevLett.69.2881} or measurement-based quantum computation \cite{PhysRevLett.86.5188}. In all these schemes entanglement plays the role of a resource that is ``consumed" during the actual implementation of the protocol. For example, in measurement-based quantum computation, a highly multi-partite entangled state (the cluster state) is adaptively measured locally until the desired unitary evolution is achieved; the cluster state model is equivalent to the well known circuit model \cite{NielsenChuang:QuantumComputation}, but, in contrast to the latter, the role of entanglement as a computational resource is now clearly visible.

Decoherence \cite{RevModPhys.75.715} constitutes a major impediment in the physical implementation of quantum information protocols. In general, entangled states are not resilient against noise, and this may first look like an impassible barrier for the construction of a working quantum computer. Quantum error correcting codes can make quantum computation possible provided the level of noise is below a certain threshold, but the current technology is far from achieving it. It is therefore of crucial importance to understand how entanglement behaves under decoherence.

Previous work \cite{Horodecki:2010:QCS:2011451.2011452,Konrad:2008fk, PhysRevLett.101.170502, PhysRevLett.105.190504} addressed this important problem of entanglement evolution (decay), however the noise model was restricted to one-sided channels: i.e. noise that acts non-trivially only on one subsystem, without affecting the other remaining parts of it.  
 
In this article we investigate multi-partite entanglement evolution under a far more general setup: the locally-correlated noisy channels: technically separable operations, which includes the LOCC as a strict subset \cite{PhysRevA.59.1070}. Separable operations (and automatically LOCC) constitute a much richer class and encompasses the previous investigated one-sided channels as well as  the usual independent noise model assumed in the theory of quantum error correction. The study of entanglement evolution under such noise model is therefore deepening our understanding of the fascinating and not-at-all well understood subject of multi-partite entanglement \cite{RevModPhys.81.865}. 

Our most surprising result is that entanglement evolution is \emph{independent} of the initial state and the entanglement measure used to quantify its evolution, and depends solely on the intrinsic properties of the noise. We provide a closed-form formula for how entanglement decays which is straightforward to calculate given the Kraus operator-sum representation of the process. 

Recent attempts to characterize entanglement evolution under locally-correlated noise, not just one-sided, were made in \cite{0953-4075-45-3-035501}, however the authors restricted their study to tripartite qubit entanglement and particular form LOCC protocols. Our result is the first one that completely characterizes the average entanglement evolution under the very general class of separable operations and may be of interest from an experimental point of view.

%\section{Entanglement evolution}\label{sct2}
\textit{Entanglement evolution.}
We consider a multipartite system of $n$ qudits, described by a Hilbert space $\HC=\HC_1\otimes\HC_2\otimes\cdots\otimes\HC_n$, and let $d_i$ be the dimension of the $i$-th local Hilbert space. Denote by $\BC(\HC)$ the set of all bounded operators (e.g., density matrices) acting on $\HC$. We further define the special linear group $G\equiv SL(d_1,\CPLX)\otimes SL(d_2,\CPLX)\otimes\cdots\otimes SL(d_n,\CPLX)$, where $SL(d,\CPLX)$ is the group of $d\times d$ complex matrices of determinant 1.

\begin{definition}\label{dfn1}
An SL-invariant multi-partite entanglement measure \cite{PhysRevLett.105.190504}, $E_{inv}(\cdot)$, is a non-zero function initially defined from pure states in $\HC$ to non-negative real numbers then extended to mixed states satisfying:
\begin{enumerate}
	\item[i)] It is SL-invariant, i.e. $E_{inv}(g\dya{\psi}g^\dagger)=E_{inv}(\dya{\psi})$, for all $g\in G$ and all $\ket{\psi}\in\HC$; 
	\item[ii)] It is homogeneous of degree 1, i.e. $E_{inv}(r\dya{\psi})=rE_{inv}(\dya{\psi})$, for all non-negative $r$ and all $\ket{\psi}\in\HC$;
	\item[iii)]  Its convex roof extension to mixed states is given by
\begin{equation}\label{eqn1}
	E_{inv}(\rho):=\min\sum_{i}p_i E_{inv}(\dya{\psi_i}), 
\end{equation}
where the minimum is taken over all possible pure states ensembles that decompose $\rho$, i.e. \mbox{$\rho=\sum_ip_i\dya{\psi_i}$}. 
\end{enumerate}
\end{definition}
The criteria above guarantee that $E_{inv}$ is an entanglement monotone under LOCC \cite{PhysRevA.68.012103}. However, we will explicitly show here that $E_{inv}$ is also monotone under the more general class of separable operations (see the upper bound in Theorem~\ref{thm1}).
These measures capture genuine multipartite entanglement and are in general dimension-dependent, i.e.,  may not exist for some choices of the local dimensions $d_i$. The measures are therefore dependent on the space the states are embedded in, as we will see below.
An example is Wootters' concurrence of two qubits \cite{PhysRevLett.80.2245} and its higher dimension generalization called G-concurrence \cite{PhysRevA.71.012318}. In particular, these are the only SL-invariant measures for the bipartite case, with $d_1=d_2$. No SL-invariant measure exists for $d_1\neq d_2$, hence a simple embedding of a $2\times 2$ state into a $2\times 3$ state is not measurable. For multi-partite systems there is  more than  one measure, and their construction is based on SL-invariant polynomials \cite{PhysRevA.68.012103,NewJPhys.7.073013}. An SL-invariant polynomial, $f(\dya{\psi})$, is a polynomial in the components of the state $\ket{\psi}$, invariant under the action of the group $G$, i.e. $f(g\dya{\psi}g^\dagger)=f(\dya{\psi})$, for all $g\in G$.
Then $E_{inv}(\dya{\psi})=|f_k(\dya{\psi})|^{1/k}$ is an example of an SL-invariant measure, where $f_k$ is a homogeneous SL-invariant polynomial of degree $k$. In particular for 3 qubits the square root of the 3-tangle \cite{PhysRevA.61.052306} is the simplest multi-partite SL-invariant measure.

In the next Lemma we prove that the convex roof extension on an SL-invariant measure satisfies the SL-invariance and homogeneity properties, similarly to the pure state case. This is the most important technical result of our paper and is crucial for the proof of our main result, Theorem~\ref{thm1}. This important Lemma was overlooked by the community, a partial version being stated before in \cite{Horodecki:2010:QCS:2011451.2011452}, with inequality instead of equality and valid only for bipartite separable operations.
\begin{lemma}\label{lma1}
Any entanglement measure defined as a convex roof extension of an SL-invariant pure state entanglement measure remains SL-invariant and homogeneous of degree 1, i.e.
\begin{enumerate}
	\item[i)] $E_{inv}(g\rho g^\dagger)=E_{inv}(\rho)$, for all $g\in G$ and all $\rho\in\BC(\HC)$; 
	\item[ii)] $E_{inv}(r\rho)=rE_{inv}(\rho)$, for all non-negative $r$ and all $\rho\in\BC(\HC)$;
\end{enumerate}
\end{lemma}
\begin{proof}
Proof of i). Let $g\in G$ and let $\{\dya{\tilde\psi_i}\}$ be an optimal un-normalized ensemble decomposition of $\rho$, i.e.
\begin{equation}\label{eqn2}
\einv(\rho)=\sum_i \einv(\dya{\tilde\psi_i}).
\end{equation}
Then $\{g\dya{\tilde\psi_i}g^\dagger\}$ must be some ensemble decomposition of the positive operator $g\rho g^\dagger$, in general not the optimal one. Hence 
\begin{align}\label{eqn3}
\einv&(g\rho g^\dagger)\leqslant \sum_i \einv(g\dya{\tilde\psi_i}g^\dagger)\notag\\
&=\sum_i \einv(\dya{\tilde\psi_i})=\einv(\rho).
\end{align}
Let now $\{\dya{\tilde\phi_k}\}$ be an optimal un-normalized ensemble decomposition of $g\rho g^\dagger$, i.e.
\begin{equation}\label{eqn4}
\einv(g\rho g^\dagger)=\sum_k \einv(\dya{\tilde\phi_k}).
\end{equation}
Note that $\{g^{-1}\dya{\tilde\phi_k}{g^{-1}}^{\dagger}\}$ is some ensemble decomposition of $\rho$, again not necessarily the optimal one, where $g^{-1}$ is the inverse of $g$ (the inverse is guaranteed to exist since $g$ has determinant 1).
Then
\begin{align}\label{eqn5}
\einv&(\rho)\leqslant \sum_k \einv(g^{-1}\dya{\tilde\phi_k}{g^{-1}}^{\dagger})\notag\\
&=\sum_k \einv(\dya{\tilde\phi_k})=\einv(g\rho g^\dagger).
\end{align}
Now part i) follows from \eqref{eqn3} and \eqref{eqn5}.

Proof of ii). This follows at once, since if $\{\dya{\tilde\psi_i}\}$ is an optimal ensemble decomposition of a density operator $\rho$, then $\{r\dya{\tilde\psi_i}\}$ will automatically be an optimal decomposition of the operator $r\rho$. Then the result follows from the homogeneity of $\einv(\cdot)$ for pure states.
\end{proof}

\begin{definition}\label{dfn2}
Let $\Lambda$ be a completely positive trace preserving (CPTP) multi-partite separable operation acting on \mbox{$\BC(\HC_1)\otimes\BC(\HC_2)\otimes\cdots\otimes\BC(\HC_n)$}, with operator-sum representation given by Kraus operators \mbox{$\{K^{(1)}_m\otimes K^{(2)}_m\otimes\cdots \otimes K^{(n)}_m\}_{m}$}, with 
\mbox{$\sum_{m} {K_m^{(1)}}^\dagger {K_m^{(1)}}\otimes \cdots \otimes {K_m^{(n)}}^\dagger {K_m^{(n)}} = I\otimes \cdots \otimes I$}. We define the \emph{entanglement resilience factor} (ERF) of $\Lambda$ to be 
\begin{equation}\label{eqn6}
\FC(\Lambda):=\min\sum_{m}|\det(K^{(1)}_m)|^{2/d_1}\cdots|\det(K^{(n)}_m)|^{2/d_n},
\end{equation}
where the minimum is taken over all separable Kraus decompositions \mbox{$\{K^{(1)}_m\otimes K^{(2)}_m\otimes\cdots \otimes K^{(n)}_m\}_{m}$} of $\Lambda$.
\end{definition}
\begin{lemma}\label{lma2}
$\FC(\Lambda)\leqslant 1$, with equality only if $\Lambda$ is a separable random unitary channel, i.e. all Kraus operators are proportional to unitaries.
\end{lemma}
\begin{proof}
From the closure condition we have
\begin{equation}\label{eqn7}
\sum_{m} {K_m^{(1)}}^\dagger {K_m^{(1)}}\otimes \cdots \otimes {K_m^{(n)}}^\dagger {K_m^{(n)}} = I\otimes \cdots \otimes I.
\end{equation}
Taking the determinant on both sides and applying
the Minkowski's inequality for a sum of positive (semidefinite) operators (see p. 47 of \cite{Bhatia:MatrixAnalysis}) yields 
\begin{equation}\label{eqn8}
\sum_{m}|\det(K^{(1)}_m)|^{2/d_1}\cdots|\det(K^{(n)}_m)|^{2/d_n}\leqslant 1,
\end{equation}
where we have used that $\det(AA^\dagger)=|\det(A)|^2$ and $\det (A\otimes B)=\det(A)^b\det(B)^a$ for $A$ and $B$ square matrices of dimension $a\times a$ and $b \times b$, respectively.
The result then follows at once from \eqref{eqn6} and \eqref{eqn8}.
\end{proof}
The ERF was first introduced in \cite{PhysRevLett.105.190504} for the particular case of one-sided operations of the form $\Lambda^{(1)}\otimes I\otimes\cdots\otimes I$, where its computation reduced to an optimization problem due to the unitary freedom in the Kraus representation of a quantum operation \cite{NielsenChuang:QuantumComputation}. However, in our case, it is not at all clear whether a separable operation admits more than one separable representation (up to a relabelling of the Kraus operators), excluding the trivial cases where the operation is a tensor product of independent channels $\Lambda=\Lambda^{(1)}\otimes\cdots\otimes\Lambda^{(n)}$. For this latter case 
\begin{equation}\label{eqn9}
\FC(\Lambda^{(1)}\otimes\cdots\otimes\Lambda^{(n)}) \leqslant \prod_k \FC(\Lambda^{(k)}),
\end{equation}
since among all possible separable representations of $\Lambda$ there are Kraus representations of the form \mbox{$\{K^{(1)}_{j_1}\otimes\cdots\otimes K^{(n)}_{j_n}\}_{j_1,\ldots,j_n}$}, with $\sum_{j_k}{K^{(k)}_{j_k}}^{\dagger} K^{(k)}_{j_k}=I$ for any individual party $k$, and this implies that the minimization \eqref{eqn6} can split into $n$ individual parts, which proves \eqref{eqn9}.

The following Theorem constitutes our central result.
\begin{theorem}[Entanglement evolution]
\label{thm1}
Let $\Lambda$ be a multi-partite separable operation with operator-sum representation given by Kraus operators \mbox{$\{K^{(1)}_m\otimes K^{(2)}_m\otimes\cdots \otimes K^{(n)}_m\}_{m}$}, with 
\mbox{$\sum_{m} {K_m^{(1)}}^\dagger {K_m^{(1)}}\otimes \cdots \otimes {K_m^{(n)}}^\dagger {K_m^{(n)}} = I\otimes \cdots \otimes I$}, let $E_{inv}(\cdot)$ be an SL-invariant entanglement measure, and let $\rho\in\BC(\HC)$ be an entangled input state with $E_{inv}(\rho)\neq 0$. Then the average output entanglement is independent of the input state and of the entanglement measure, and is given by
\begin{align}\label{eqn10}
&\frac{\sum_{m}p_mE_{inv}(\sigma_m)}{E_{inv}(\rho)}=\notag\\
&\quad\sum_{m}|\det(K^{(1)}_m)|^{2/d_1}\cdots|\det(K^{(n)}_m)|^{2/d_n}\leqslant 1,
\end{align}
with equality if and only if $\Lambda$ is a separable random unitary channel \cite{1367-2630-10-2-023011}. Here  \mbox{$p_m\sigma_m=(K^{(1)}_m\otimes \cdots \otimes K^{(n)}_m) \rho (K^{(1)}_m\otimes \cdots \otimes K^{(n)}_m)^\dagger$} and $\sum_m p_m=1$.
\end{theorem}
In particular note that the upper bound in \eqref{eqn10} explicitly shows that $E_{inv}$ is an entanglement monotone under separable operations, as we claimed in the remarks following Definition~\ref{dfn2}.
\begin{proof}
We start by noting that
\begin{equation}
\label{eqn11}
(K^{(1)}_m\otimes \cdots \otimes K^{(n)}_m) \rho (K^{(1)}_m\otimes \cdots \otimes K^{(n)}_m)^\dagger=p_m\sigma_m.
\end{equation}
For the Kraus operators that have non-vanishing determinant (i.e. are invertible), we have
%\begin{widetext}
\begin{align}
|&\det (K^{(1)}_m)|^{2/d_1}\cdots |\det (K^{(n)}_m)|^{2/d_n}\notag\\
&\times\left[\frac{K^{(1)}_m\otimes \cdots \otimes K^{(n)}_m}{(\det K^{(1)}_m)^{1/d_1}\cdots (\det K^{(n)}_m)^{1/d_n}}\right](\rho)\notag\\
&\times\left[\frac{K^{(1)}_m\otimes \cdots \otimes K^{(n)}_m}{(\det K^{(1)}_m)^{1/d_1}\cdots (\det K^{(n)}_m)^{1/d_n}}\right]^\dagger=p_m\sigma_m.
\label{eqn12}
\end{align}
%\end{widetext}
Note that the quantity in square brackets has determinant of magnitude one. Use now the homogeneity and SL-invariance property  of any convex roof SL-invariant measure (see Lemma~\ref{lma1}) and apply $E_{inv}$ to \eqref{eqn12} to get
\begin{equation}
\label{eqn13}
|\det (K^{(1)}_m)|^{2/d_1}\cdots |\det (K^{(n)}_m)|^{2/d_n}E_{inv}(\rho)=p_mE_{inv}(\sigma_m).
\end{equation}
The Kraus operators that have determinant zero produce states with $p_mE_{inv}(\sigma_m)=0$, see Lemma~1 of \cite{PhysRevLett.105.190504}, and hence \eqref{eqn13} is generally valid for all Kraus operators.  Now summing \eqref{eqn13} over $m$ yields the desired equality.
The upper bound of 1 follows again from Minkowski's inequality for a sum of positive semidefinite operators (see the proof of Lemma~\ref{lma2}).
\end{proof}
Operationally Theorem~\ref{thm1} can be thought of in terms of a separable measurement repeated multiple times, each time the outcome $m$ and the resulting state $\sigma_m$ being recorded and at the end the output entanglement being averaged over all outcomes. Theorem~\ref{thm1} implies that the ratio between the average output entanglement and the initial entanglement can be computed \emph{explicitly} as a function of the operator-sum representation of the separable operation and has the following properties: i) it is independent of the input state; and ii) it is independent of the entanglement measure, as long as the latter is SL-invariant.
%\begin{enumerate}
%\item it is independent of the input state
%\item it is independent of the entanglement measure
%\item it is invariant under embedding into higher dimensional systems by addition of an extra qudit, i.e. replace $\Lambda$ by $\Lambda\otimes I$ and $\rho$ by $\rho\otimes \rho_0$, where $\rho_0$ describes the state of the additional qudit.
%\end{enumerate}

In Theorem~\ref{thm1} we considered only input states with non-vanishing entanglement, since whenever $E_{inv}(\rho)=0$ the average output entanglement must also be zero, hence such non-entangled states $\rho$ are not of interest. 

The following Corollary then follows.
\begin{corollary}\label{corl1}
The ERF is bounded by
\begin{equation}\label{eqn14}
\frac{E_{inv}(\Lambda(\rho))}{E_{inv}(\rho)}\leqslant \FC[\Lambda]\leqslant   \frac{\sum_{m}p_mE_{inv}(\sigma_m)}{E_{inv}(\rho)}.
\end{equation}
\end{corollary}

\begin{proof}
The inequality \mbox{$\FC(\Lambda)\leqslant (\sum_{m}p_mE_{inv}(\sigma_m))/E_{inv}(\rho)$} follows from the definition of the ERF as a minimum over all Kraus decompositions, see  \eqref{eqn6}, and the equality \eqref{eqn13} proved above.

Finally, the inequality $E_{inv}(\Lambda(\rho))/E_{inv}(\rho)\leqslant\FC(\Lambda)$ can be proved by noting that $E_{inv}(\Lambda(\rho))/E_{inv}(\rho)$ is independent of the Kraus representation of the channel. We can therefore choose the Kraus representation that achieves the minimum in \eqref{eqn6}, \mbox{$\{\tilde K^{(1)}_m\otimes \tilde K^{(2)}_m\otimes\cdots \otimes \tilde K^{(n)}_m\}_{m}$}, which yields
\begin{equation}\label{eqn15}
\FC(\Lambda)=\sum_{m}|\det(\tilde K^{(1)}_m)|^{2/d_1}\cdots|\det(\tilde K^{(n)}_m)|^{2/d_n}.
\end{equation}
But 
\begin{equation}\label{eqn16}
\sum_{m}|\det(\tilde K^{(1)}_m)|^{2/d_1}\cdots|\det(\tilde K^{(n)}_m)|^{2/d_n}=\frac{\sum_{m}p_mE_{inv}(\sigma_m)}{E_{inv}(\rho)},
\end{equation}
as we just showed. Now, we use the convexity of $E_{inv}$ (we remind the reader that the measure $E_{inv}$ is defined on mixed states via a convex roof extension, see part iii) of Definition~\ref{dfn1}), to get
\begin{equation}\label{eqn17}
\frac{E_{inv}(\Lambda(\rho))}{E_{inv}(\rho)}\leqslant \frac{\sum_{m}p_mE_{inv}(\sigma_m)}{E_{inv}(\rho)}=\FC(\Lambda),
\end{equation}
and this concludes the proof.
\end{proof}
In particular, for one-sided operations, $\FC(\Lambda)$ can be computed explicitly \cite{PhysRevLett.105.190504} and, additionally, for one-sided operations and pure state inputs, the first inequality of \eqref{eqn14} is actually an equality, $E_{inv}(\Lambda(\dya{\psi}))/E_{inv}(\dya{\psi})=\FC(\Lambda)$ \cite{PhysRevLett.105.190504}. However, in the most general case of multi-partite separable operations, this is no longer true. Consider for example a bipartite separable operation $\Lambda$ with an operator-sum representation given by $\{\sqrt{p}I\otimes I, \sqrt{1-p}X\otimes X\}$, where $0<p<1$ and $X=\dyad{0}{1}+\dyad{1}{0}$ is the usual bit-flip Pauli operator. In this case it is not hard to see that there are no other separable representations of $\Lambda$ (any two Kraus representation are related by a unitary matrix, and it follows at once that a different representation of $\Lambda$ will not be separable). The ERF $\FC[\Lambda]$ can then be computed explicitly, since there is no minimization to be done, and equals 1. The maximally entangled state $\ket{\psi}=(\ket{00}+\ket{11})/\sqrt{2}$ is invariant under $\Lambda$, hence 
$E_{inv}(\Lambda(\dya{\psi}))/E_{inv}(\dya{\psi})=\FC[\Lambda]=1$ and the bound is saturated. 
But in general an arbitrary bipartite state $\ket{\phi}$ is not invariant under $\Lambda$ and has $E_{inv}(\Lambda(\dya{\phi}))/E_{inv}(\dya{\phi})<1$, as one can easily check using the closed form expression for the bipartite concurrence \cite{PhysRevLett.80.2245}.

\textit{Entanglement breaking.} 
Interestingly, the capability of a channel to destroy entanglement is another quantity that is also independent of the initial state, as we will explain below. This is closely related to the work of \cite{PhysRevA.72.032317}, in which the authors define correlation measures (classical and quantum) based on how much ``noise" one has to ``inject" by local operations into an entangled quantum state to make it separable, in the limit of many copies and asymptotically vanishing errors. We now point out that such a measure can not be defined in the zero-error single copy case, since the quantity of interest is state-independent. The capability of a channel to destroy entanglement is quantified as follows.
\begin{definition}
A CPTP map \mbox{$\Phi:\BC(\HC_d)\longrightarrow\BC(\HC_d)$}  is called $r$-partially entanglement breaking \cite{Chruscinski:2006:PEB:1120154.1120156} if
\begin{equation}\label{eqn18}
SN\left[(\Phi\otimes I)\rho\right]\leqslant r,\quad\forall\rho\in\BC(\HC_d\otimes\HC_d).
\end{equation}
\end{definition}
It generalizes the usual notion of entanglement breaking channels \cite{RevMathPhys15.629}, the latter having $r=1$. Here $SN(\sigma)$ denotes the Schmidt number  \cite{PhysRevA.61.040301} of the bipartite density operator $\sigma$, defined as
\begin{equation}\label{eqn19}
SN(\rho):=\min_{\{p_i,\ket{\psi_i}\}}\left\{\max_i SR(\ket{\psi_i}) \right\},
\end{equation}
where the minimum is taken over all possible pure state ensemble decompositions of $\rho$ and
$SR(\ket{\psi_i})$ denotes the Schmidt rank on $\ket{\psi_i}$, i.e. the number of non-zero Schmidt coefficients of $\ket{\psi_i}$. A state $\sigma$ is separable if and only if $SN(\sigma)=1$. 

The following proposition was recently proved in \cite{quantph.1110.4363} (see Proposition 8).
\begin{proposition}
A CPTP map \mbox{$\Phi:\BC(\HC_d)\longrightarrow\BC(\HC_d)$}  is $r$-partially entanglement breaking if and only if there exist a full Schmidt rank pure state $\ket{\psi}\in\HC_d\otimes\HC_d$ (i.e. having all Schmidt coefficients strictly positive) for which
\begin{equation}\label{eqn20}
SN\left[(\Phi\otimes I)\dya{\psi}\right]=r.
\end{equation}
\end{proposition}
In other words, if a channel breaks entanglement for one (maximal Schmidt number) state, then it must break entanglement for \emph{all} states, that is, its ``disentangling power" is an intrinsic property of the channel independent of the input state, as long as the input has maximal Schmidt number (this corresponds to the case of non-zero entanglement as quantified by any SL-invariant measure). This means that one can not use the approach of \cite{PhysRevA.72.032317} to quantify entanglement (or correlations) as the total amount of local noise one has to inject in a state to make it separable, since in a single copy zero-error regime this quantity is state independent.

%\section{Conclusions and open problems}\label{sct4}
\textit{Conclusions and open questions.}
We studied the evolution of mixed state entanglement under local decoherence. Our main result is summarized by Theorem~\ref{thm1}, in which we proved that multi-partite mixed-state entanglement evolution under separable operations is measure and state independent and is solely a function of the channel. This significantly extends the results in \cite{Horodecki:2010:QCS:2011451.2011452,Konrad:2008fk,PhysRevLett.101.170502} applicable only for the restricted case of bipartite separable operations, and also generalize the previous work of \cite{PhysRevLett.105.190504}, where only one-sided operations of the form $\Lambda^{(1)}\otimes I\otimes\cdots\otimes I$ were considered. For the first time we derived a closed form expression for the average entanglement evolution under the most general class of separable channels, and our results are automatically valid for LOCC since the later is a proper subset of the former.

An interesting fact is that our equality in Theorem~\ref{thm1} depends on the Kraus representation of the separable operation. Given such a separable operation, are there more than one separable representations of it (aside from a permutation of the Kraus operators)? A complete solution to this problem will allow an explicit evaluation of the ERF $\FC(\Lambda)$ and provide a non-trivial upper bound on the entanglement evolution of arbitrary mixed states, Corollary~\ref{corl1}. We observed numerically using non-linear optimization methods that there are many cases of separable operations that do not admit more than one nontrivial Kraus representation, but this is far from a rigorous proof.

Our results are valid for SL-invariant entanglement measures based on convex roof extension and are not applicable in general to arbitrary entanglement measures. However, it follows at once that the convex roof extension is an upper bound on \emph{any} other convex measure that coincides with the former on pure states, i.e. both give the same result on a pure state. In this case Theorem~\ref{thm1}  provides a useful upper bound on the average entanglement evolution, i.e. the equality in \eqref{eqn10} becomes an inequality. 

We discussed the ``disentangling power" of a quantum channel and pointed out that it exhibits a similar state-invariance as our average entanglement ratio in Theorem~\ref{thm1}. This forbids a straightforward extension of \cite{PhysRevA.72.032317} to the single copy zero-error regime. An interesting generalization will be to define a disentangling power of a separable operation, instead on a single-sided channel, and characterize the class of $r$-partially entanglement breaking such operations. This approach may prove useful in extending \cite{PhysRevA.72.032317} to the single copy zero-error regime, since one does not expect the disentangling power to be state-invariant anymore.

\section*{Acknowledgments}
We acknowledge support from the Natural Sciences and Engineering Research Council (NSERC) of Canada and from Pacific Institute for Mathematical Sciences (PIMS).

%\bibliography{../../bibliography/mybib}

%merlin.mbs 2010-03-15 4.21a (PWD, AO, DPC)
%Control: key (0)
%Control: author (8) initials jnrlst
%Control: editor formatted (1) identically to author
%Control: production of article title (-1) disabled
%Control: page (0) single
%Control: year (1) truncated
%Control: production of eprint (0) enabled
%

\end{document}